\documentclass[
aip,jmp,
reprint,
a4paper,
longbibliography,
nofootinbib,
nobibnotes,
]{revtex4-1}
\usepackage[utf8]{inputenc}
\usepackage[T1]{fontenc}

\usepackage{amsmath,amssymb,amsthm,amsfonts,mathtools}
\usepackage{dsfont}
\usepackage{braket}
\usepackage{hyperref}
\usepackage{enumitem}

\hypersetup{citecolor=magenta}
\hypersetup{colorlinks=true}
\hypersetup{linkcolor=blue}
\hypersetup{urlcolor=blue}

\theoremstyle{definition}
\newtheorem{theorem}{Theorem}
\newtheorem{lemma}[theorem]{Lemma}

\let\openone\undefined
\DeclareMathOperator\AHerm{{\mathbf A}}
\DeclareMathOperator\Herm{{\mathbf H}}
\DeclareMathOperator\Mat{Mat}

\DeclareMathOperator\tr{tr}
\DeclarePairedDelimiter\abs\lvert\rvert
\DeclarePairedDelimiter\exv\langle\rangle
\DeclarePairedDelimiter\norm\lVert\rVert
\newcommand\ee{\mathrm{e}}
\newcommand\dd{\mathrm{d}}
\newcommand\reals{\mathds{R}}
\newcommand\compl{\mathds{C}}
\newcommand\quats{\mathds{H}}
\newcommand\openone{\mathds{1}}
\newenvironment{myenum}
{\begin{itemize}[itemsep=0pt,leftmargin=4em]}{\end{itemize}}

\begin{document}

\title{Quaternionic quantum theory admits universal dynamics only for two-level 
systems}

\author{Jonathan Steinberg}
\email{jonathan.steinberg@student.uni-siegen.de}
\author{H.\ Chau Nguyen}
\email{chau.nguyen@uni-siegen.de}
\author{Matthias Kleinmann}
\email{matthias.kleinmann@uni-siegen.de}
\affiliation{Naturwissenschaftlich--Technische Fakultät,
Universität Siegen, 57068 Siegen, Germany}

\begin{abstract}
We revisit the formulation of quantum mechanics over the quaternions and 
 investigate the dynamical structure within this framework.
Similar to standard complex quantum mechanics, time evolution is then mediated 
 by a unitary operator which can be written as the exponential of the generator 
 of time shifts.
By imposing physical assumptions on the correspondence between the energy 
 observable and the generator of time shifts, we prove that quaternionic 
 quantum theory admits a time evolution only for systems with a quaternionic 
 dimension of at most two.
Applying the same strategy to standard complex quantum theory, we reproduce 
 that the correspondence dictated by the Schrödinger equation is the only 
 possible choice, up to a shift of the global phase.
\end{abstract}

\maketitle

%%%%%%%%%%%%%%%%%%%%%%%%%%%%%%%%%%%%%%%%%%%%%%%%%%%%%%%%%%%%%%%%%%%%%%%%%%%%
\section{Introduction}
%%%%%%%%%%%%%%%%%%%%%%%%%%%%%%%%%%%%%%%%%%%%%%%%%%%%%%%%%%%%%%%%%%%%%%%%%%%%
Our understanding of quantum theory has significantly improved by investigating 
 alternatives to quantum theory and analyzing how these alternatives would or 
 would not be at variance with observations or with our expectations on the 
 structure of a physical theory.
Recently, these investigations are often based on the formalism of generalized 
 probabilistic theories, where the fundamental objects are the convex sets of 
 states and measurements.
Different sets of assumptions have been found which are sufficient to single 
 out quantum theory as the only possible theory \cite{Hardy2001, 
 Chiribella2011, Dakic2011, Masanes2011, Hardy2011, Wilce2016, Hohn2017}.
A special role in this set of assumptions plays the analysis of the dynamics of 
 such generalizations of quantum theory, see, for example, 
 Ref.~\onlinecite{Barnum2014}.
Specifically, in quantum mechanics (and also in classical mechanics) there is 
 an intimate relation between the Hamiltonian $H$ as the energy observable and 
 the generator of time shifts $-\frac i\hbar H$ as it occurs in the Schrödinger 
 equation.

Maybe the most notable early alternatives to quantum mechanics that have been 
 studied in great detail are real and quaternionic quantum mechanics.
Those are based on the question, why the wave function in quantum theory is 
 complex valued and how quantum theory would change when real valued or 
 quaternionic valued wave functions are used.
The two main concerns for the real and quaternionic case are the composition of 
 systems via a tensor product and a suitable modification of the Schrödinger 
 equation.
For real quantum theory, both topics lead to basically the same conclusion, 
 namely that there must be a superselection rule \cite{Stueckelberg1960, 
 McKague2009, Aleksandrova2013}.
In quaternionic quantum theory, it difficult to obtain any sensible notion of a 
 tensor product \cite{Horwitz1984, Baez2012, Joyce2001, Ng2007} --- however, 
 the need for composing systems can also be questioned
 \cite{Chiribella2018}.
A consistent dynamics in quaternionic quantum theory has been formulated, 
 however, also at the price of a superselection rule, where only a subspace of 
 all self-adjoint operators can be used as the Hamiltonian of a system 
 \cite{Finkelstein1962}.
A further grain of salt for real and quaternionic quantum theory is that they 
 have been noted to reduce to the complex case under Poincaré symmetry 
 \cite{Moretti2017, Moretti2019}.
Despite of these concerns, Peres \cite{Peres1979} suggested a possible 
 experiment on the basis of noncommuting phases which could reveal the 
 characteristics of quaternionic quantum mechanics.
Corresponding experiments were realized using neutron interferometry 
 \cite{Kaiser1984} and using single photon interferometry \cite{Procopio2017, 
 Adler2017, Procopio2017b}.
An alternative method to test for hypercomplex theories deals with transmission 
 and reflection times in presence of a quaternionic potential 
 step\cite{DeLeo2006}.

In this paper we use the standard quaternionic formulation \cite{Adler1995} of 
 states and observables, that is, states are normalized vectors and observables 
 are self-adjoint matrices over the quaternions.
We ask which dynamical evolution is admissible in this case.
For complex quantum mechanics, the Schrödinger equation implies that the state 
 evolves according to a unitary group parametrized by time.
Each such group is determined by the Hamiltonian of the system.
We seek for a similar construction with the aim to derive a Schrödinger-type 
 equation for quaternionic quantum theory.
In contrast to previous work \cite{Finkelstein1962, Adler1995}, we are 
 interested in the universal case where the set of Hamiltonians is 
 unrestricted, that is, every self-adjoint operator must induce some dynamics.
We find that this is only possible for one-level or two-level systems and that 
 the corresponding Schrödinger-type equation is necessarily of the form
\begin{equation}
 \hbar\dot\psi(t) = [AH+HA-\tr(H)A]\psi(t),
\end{equation}
 where $H$ is the Hamiltonian and $A$ is a skew-adjoint operator which is 
independent of $H$.
The term in the square brackets replaces here $-iH$ from the Schrödinger 
 equation.
We arrive at this result assuming that the term in the square brackets is an 
 $\reals$-linear expression in $H$ and that it commutes with $H$.

The paper is organized as follows.
In Section~\ref{sec:canonicalqt} we consider the case of standard complex 
 quantum mechanics.
We review the connection between the Schrödinger equation, generators of time 
 shifts, and Stone's theorem and develop then the axioms for the correspondence 
 between the Hamiltonian and the generator of time shifts.
In Theorem~\ref{thm:cphi} we establish that these axioms are sufficient to 
 reproduce the Schrödinger equation.
In Section~\ref{sec:quaternionicqt} we then turn to the quaternionic case.
We first summarize the mathematical preliminaries to formulate quantum theory 
 over the quaternions.
Our main result, the characterization of all possible dynamics in quaternionic 
 quantum theory is provided in Theorem~\ref{thm:qrep}.
Subsequently we discuss alternatives to the axioms that lead to 
 Theorem~\ref{thm:qrep} before we conclude in Section~\ref{sec:conclusions}.

%%%%%%%%%%%%%%%%%%%%%%%%%%%%%%%%%%%%%%%%%%%%%%%%%%%%%%%%%%%%%%%%%%%%%%%%%%%%
\section{Time evolution in complex quantum theory}\label{sec:canonicalqt}
%%%%%%%%%%%%%%%%%%%%%%%%%%%%%%%%%%%%%%%%%%%%%%%%%%%%%%%%%%%%%%%%%%%%%%%%%%%%
In quantum mechanics, the time evolution of a system is described by the 
 Schrödinger equation, $i\hbar \dot\psi(t) = H \psi(t)$.
For a time-independent Hamiltonian this gives rise to the unitary time 
 evolution operator
\begin{equation}\label{eq:ust}
 U_t = \ee^{-\frac i\hbar Ht},
\end{equation}
 which provides the solution of the Schrödinger equation via $\psi(t_0+t)= U_t 
 \psi(t_0)$.
From a fundamental perspective it is more natural to first require that the 
 dynamics should be governed by some unitary operators parametrized by $t$ and 
 then conclude that the Schrödinger equation must hold.
This approach is formalized by Stone's representation theorem of unitary 
 groups.

%%%%%%%%%%%%%%%%%%%%%%%%%%%%%%%%%%%%%%%%%%%%%%%%%%%%%%%%%%%%%%%%%%%%%%%%%%%%
\subsection{Unitary groups and Stone's theorem}
%%%%%%%%%%%%%%%%%%%%%%%%%%%%%%%%%%%%%%%%%%%%%%%%%%%%%%%%%%%%%%%%%%%%%%%%%%%%
We aim to obtain an evolution operator $U_t\colon \psi(t_0)\mapsto \psi(t_0+t)$ 
 of the same structure as in Eq.~\eqref{eq:ust}, but without building on the 
 Schrödinger equation.
For this we assume that $U_t$ is linear in $\psi(t_0)$, preserves the norm of 
 $\psi(t_0)$, is independent of $t_0$, and is continuous in $t$.
More precisely, $(U_t)_{t\in \reals}$ is a strongly continuous unitary group, 
 that is,
\begin{myenum}
\item[(i)] $U_t$ is a unitary operator for all $t$, with $U_0=\openone$,
\item[(ii)] $U_{s+t}= U_s U_t$ for all $s$, $t$, and
\item[(iii)] $t\mapsto U_t$ is strongly continuous.
\end{myenum}
Condition (i) expresses linearity and isometry and $U_0=\openone$ is used to 
 implement the identity $\psi(t_0)= U_0\psi(t_0+0)$.
Condition (ii) follows from the independence of $t_0$ and condition (iii) is a 
 specification of the assumption of continuity.
Strong continuity refers to the strong operator topology and reduces because of 
 (i) and (ii) to $\lim\limits_{t\to 0} \norm{U_t\psi- \psi}= 0$ for all $\psi$.

The fundamental representation theorem of strongly continuous one-parameter 
 unitary groups is due to Stone (see, for example, 
 Ref.~\onlinecite{Teschl2009}).
For every such group $(U_t)_{t\in\reals}$ there exists a unique skew-adjoint 
 operator $A=-A^\dag$, such that
\begin{equation}\label{eq:cstone}
 U_t = \ee^{At}
\end{equation}
 holds.
This result is valid for general Hilbert spaces, with subtleties occurring if 
 $t\mapsto U_t\psi$ is not differentiable at $t=0$ for some $\psi$.
Clearly, in the finite-dimensional case this function is always differentiable.

%%%%%%%%%%%%%%%%%%%%%%%%%%%%%%%%%%%%%%%%%%%%%%%%%%%%%%%%%%%%%%%%%%%%%%%%%%%%
\subsection{Hamiltonians and generators of time shifts}\label{sec:cgens}
%%%%%%%%%%%%%%%%%%%%%%%%%%%%%%%%%%%%%%%%%%%%%%%%%%%%%%%%%%%%%%%%%%%%%%%%%%%%
From a physical perspective, the skew-adjoint operator $A$ in 
 Eq.~\eqref{eq:cstone} is responsible for translations of time and thus it can 
 be seen as the generator of time shifts.
The Schrödinger equation implies that the correspondence between the 
 Hamiltonian and the generator of time shifts is given by $A=-\frac i\hbar H$.
But is this the only way to establish a correspondence between the generator of 
 time shifts and the Hamiltonian and if not, how can we classify the different 
 possibilities?

Note that we make here a subtle distinction between energy observable, 
 Hamiltonian, and generator of time shifts.
The skew-adjoint generator of time shifts is related to the self-adjoint 
 Hamiltonian in a way we are about to discuss.
Since the Hamiltonian is self-adjoint, it has the structure of an observable 
 and we require below that the Hamiltonian has features commonly connected to 
 the energy of a system, in particular that the Hamiltonian is independent 
 under the time evolution that it induces.
In this sense the Hamiltonian can also be seen as the energy observable of the 
 system.

To address this question, we assume some relation $\varphi$ between the 
 Hamiltonian and the generator of time shifts.
Since the former are self-adjoint, that is, $H=H^\dag$, while the latter are 
 skew-adjoint, that is, $A=-A^\dag$, we write $\varphi\subset 
 \Herm(n,\compl)\times \AHerm(n,\compl)$.
Here, $\Herm(n,\compl)$ and $\AHerm(n,\compl)$ denote the $\reals$-vector space 
 of self-adjoint and skew-adjoint complex $n\times n$ matrices, respectively.
For the Schrödinger equation we have the one-to-one correspondence 
 $H\leftrightarrow -\frac i\hbar H$ and the relation $\phi$ is then given by 
 $\varphi_S = \set{ (H,-\frac i\hbar H) | H\in \Herm(n,\compl) }$.

In literature, such a relation $\varphi$ is sometimes called a dynamical 
 correspondence \cite{Alfsen1998, Barnum2014, Branford2018} and different 
 assumptions of the properties of $\varphi$ have been discussed, see, for 
 example, Refs.~\onlinecite{Branford2018, Barnum2014}.
We choose here the following restrictions on $\varphi$ by requiring that 
 $\varphi$ is
\begin{myenum}
\item[(DC1)]
 $\reals$-homogeneous, that is, $(H,A)\in \varphi$ implies $(\lambda H,\lambda 
 A)\in\varphi$ for all $\lambda\in \reals$,
\item[(DC2)]
 additive, that is, $(H,A)\in \varphi$ and $(H',A')\in \varphi$ implies 
 $(H+H',A+A')\in \varphi$, and
\item[(DC3)]
 commuting, that is, $(H,A)\in \varphi$ implies $HA=AH$.
\end{myenum}
The relation $\varphi$ should be $\reals$-homogeneous (DC1) to match the 
 intuition that higher energies correspond in direct proportion to faster time 
 evolutions and vice versa.
In addition the correspondence in classical mechanics as well as in complex 
quantum theory has this property.
The assumption of additivity (DC2) can be easily justified for the case where 
 $H$ and $H'$ commute since in this case we have the addition of two 
 Hamiltonians that differ only in spectrum.
Then for any eigenstate of $H$ and $H'$ an argument similar to the motivation 
 for assumption (DC1) can be applied.
Noncommuting Hamiltonians most prominently appear in the form of interaction 
 Hamiltonians, where $H=H_A+H_B+\mu H_I$.
In many situations the interaction strength $\mu$ is an external experimental 
 parameter, for example, the strength of a magnetic field, and hence additivity 
 is a reasonable assumption.
This reasoning might be more difficult to apply, if the Hamiltonian does not 
 emerge as an effective description, but is an unalterable property of the 
 system.
Finally, the relation $\varphi$ should also be commuting (DC3), so that the 
 Hamiltonian is invariant under time shifts.
That is, the observable $H$ should be a conserved quantity under its own time 
 evolution,
\begin{equation}
 U_t^\dag H U_t = \ee^{At} H \ee^{-At} = H,
\end{equation}
 for all $(H, A)\in \varphi$.
The properties (DC1) and (DC2) make $\varphi$ an $\reals$-linear space and any 
 relation $\varphi$ obeys (DC1)--(DC3) if and only if the set 
 $\set{H+A|(H,A)\in \varphi}$ is an $\reals$-linear subspace of the normal 
 matrices.
Note that we do not include another familiar condition on $\varphi$, namely, 
 that it should be
\begin{myenum}
\item[(DC4)] covariant under unitary transformations,
\end{myenum}
 that is, $(H, A)\in \varphi$ implies $(V^\dag HV,V^\dag A V)\in \varphi$, for 
 any unitary operator $V$.

It is conceivable that the Hamiltonian alone does not determine the time 
 evolution completely and that the map $\varphi_1\colon H\mapsto \set{A| 
 (H,A)\in \varphi}$ is multivalued.
The set $\varphi_1(H)$ could also be empty, in which case $H$ would not 
 correspond to any dynamics, that is, $H$ would be unphysical.
Conversely, the same dynamics might arise from different Hamiltonians so that 
 $\varphi_2: A\mapsto \set{H|(H,A)\in \varphi}$ is multivalued.
Also the set $\varphi_2(A)$ could be empty and hence the corresponding time 
 evolution would be unphysical.
In this paper, we focus on the first case and we only consider maps $\Phi\colon 
 \Herm(n,\compl)\to \AHerm(n,\compl)$ such that $\varphi_1(H)= \set{\Phi(H)}$.

The conditions (DC1)--(DC3) are equivalent to the condition that $\Phi$ is a 
 commuting $\reals$-linear map, that is, $\Phi$ is $\reals$-linear with 
 $H\Phi(H)=\Phi(H)H$.
We now demonstrate that such maps must have a very specific form.
For this we use the following result by Brešar.
\begin{lemma}[Corollary 3.3 in Ref.~\onlinecite{Bresar2004}]\label{l:bresar}
Let $\mathcal A$ be a simple unital ring.
Then every commuting additive map $f$ on $ \mathcal A$ is of the form $f(x) = 
 zx + g(x)$, where $z\in \mathcal Z (\mathcal A)$ and $g\colon \mathcal A \to 
 \mathcal Z(\mathcal A)$ is an additive map.
\end{lemma}
Here, $\mathcal Z(\mathcal A)$ denotes the center of $\mathcal A$, that is, the 
 elements of $\mathcal A$ that commute with all of $\mathcal A$.
Since $\Mat(n,\compl)$ is a simple unital ring with center $\compl\openone$, we 
 can in principle apply Lemma~\ref{l:bresar} in order to characterize all 
 maps $\Phi$.
However, we first need to extend the domain of $\Phi$ to be all of 
 $\Mat(n,\compl)$.
This is readily achieved by the canonical extension $\Phi_e$ of $\Phi$ via
\begin{equation}
 \Phi_e(M_1 + iM_2)= \Phi(M_1) + i\Phi(M_2),
\end{equation}
 where $X= M_1+iM_2$ is the unique decomposition of $X\in \Mat(n,\compl)$ into 
 its self-adjoint and skew-adjoint part and $M_1$, $M_2$ are self-adjoint 
 matrices.
Clearly this extension is additive ($\compl$-homogeneous) if $\Phi$ is additive 
 ($\reals$-homogeneous).
The extension is also commuting if $\Phi$ is additive and commuting.
Indeed we have $[\Phi(M_1),M_2]+ [\Phi(M_2),M_1]= 0$ due to 
 $[\Phi(M_1+M_2),M_1+M_2]= 0$ and hence $[\Phi_e(X), X]= i([\Phi (M_1), M_2]+ 
 [\Phi(M_2),M_1])= 0$ holds.
This embedding together with Lemma~\ref{l:bresar} yields the following 
 characterization of all maps obeying (DC1)--(DC3).
\begin{theorem}\label{thm:cphi}
Every commuting $\reals$-linear map $\Phi\colon \Herm(n, \compl) \to \AHerm (n, 
 \compl)$ is of the form
\begin{equation} \label{eq:Schroed}
 \Phi(H)= i \lambda H + i\tr(B H)\openone,
\end{equation}
 where $\lambda\in \reals$ and $B\in\Herm(n,\compl)$.
\end{theorem}
\begin{proof}
Applying Lemma~\ref{l:bresar} to the extended map $\Phi_e$ yields $\Phi_e(H)= 
 \eta H+\tr(QH)\openone$ where $H\in \Herm(n,\compl)$, $\eta\in \compl$, and 
 $Q\in \Mat(n,\compl)$.
Note that we used the $\compl$-linearity of $\Phi_e$ to write $g(H)= 
 \tr(QH)\openone$.
Since $\Phi_e(H)$ must be skew-adjoint, $\eta+\eta^*=0$ and $Q+Q^\dag=0$ 
 follows, that is, $\eta\in i\reals$ and $Q\in i\Herm(n,\compl)$.
\end{proof}
Note that $\Phi(H)$ is covariant if and only if $B$ is a real multiple of 
 $\openone$.
In this case the choice of $B$ is in direct correspondence to the zero energy 
 of the system.
For any other choice of $B$ and a generic Hamiltonian the second summand of 
 Eq.~\eqref{eq:Schroed} would not transform covariant under the action of 
 unitaries and thus differs from the transformation behavior of the Hamiltonian 
 itself.
Note that the term involving $B$ cannot be measured since it would only cause a 
 global phase shift on the quantum state.
In quantum mechanics, we have $\lambda= -\frac1\hbar$ and $B=0$.
The value of $\lambda$ constitutes a constant of nature, including a sign 
 convention.

%%%%%%%%%%%%%%%%%%%%%%%%%%%%%%%%%%%%%%%%%%%%%%%%%%%%%%%%%%%%%%%%%%%%%%%%%%%%
\section{Quaternionic quantum theory}\label{sec:quaternionicqt}
%%%%%%%%%%%%%%%%%%%%%%%%%%%%%%%%%%%%%%%%%%%%%%%%%%%%%%%%%%%%%%%%%%%%%%%%%%%%
We have outlined the formalism for obtaining the Schrödinger equation for the 
 familiar case of complex quantum mechanics.
Our choices and assumptions have been made such that we can extend our 
 considerations to construct a dynamical quantum theory over the quaternions.
We start by summarizing a quaternionic version of quantum theory (see, for 
 example, Ref.~\onlinecite{Adler1995}) and we then proceed by characterizing 
 possible expressions for a Schrödinger-type equation.

%%%%%%%%%%%%%%%%%%%%%%%%%%%%%%%%%%%%%%%%%%%%%%%%%%%%%%%%%%%%%%%%%%%%%%%%%%%%
\subsection{The quaternions}
%%%%%%%%%%%%%%%%%%%%%%%%%%%%%%%%%%%%%%%%%%%%%%%%%%%%%%%%%%%%%%%%%%%%%%%%%%%%
The quaternions $\quats$ are an extension of the real and complex numbers.
They form an associative division ring where multiplication is noncommutative.
Any quaternion $q \in \quats$ can be written in the form
\begin{equation}
 q = a_1 + a_2 i + a_3 j + a_4 k,
\end{equation}
 where the coefficients $a_\ell$ are real numbers and $i,j,k$ are the 
 quaternion units, which play a role similar to the complex unit $i$.
The multiplication on $\quats$ is commutative for the real numbers and 
 otherwise determined by
\begin{equation}
 i^2 = j^2 = k^2 =-1 \quad \text{and} \quad ijk=-1,
\end{equation}
 yielding $ij=k=-ji$, $jk=i=-kj$, $ki=j=-ik$.
We identify the complex numbers as a subset of the quaternions by identifying 
 the complex unit $i$ with the quaternion unit $i$.
This allows us to write uniquely $q=a+bj$ for $a,b\in\compl$.
Conjugation is defined similarly to the complex numbers, %
\begin{equation}
 q^* = a_1- a_2i- a_3j- a_4k,
\end{equation}
 yielding the rules $(uv)^* = v^*u^*$, $qq^*=q^*q$, and $(q^*)^* = q$.
The modulus $\abs{q}=\sqrt{q q^*}$ induces the euclidean norm 
 $\norm{(a_1,a_2,a_3,a_4)}= \abs{a_1+a_2i+a_3j+a_4k}$.
This way, the quaternions are a complete normed $\reals$-algebra.

%%%%%%%%%%%%%%%%%%%%%%%%%%%%%%%%%%%%%%%%%%%%%%%%%%%%%%%%%%%%%%%%%%%%%%%%%%%%
\subsection{Modules and matrices}
%%%%%%%%%%%%%%%%%%%%%%%%%%%%%%%%%%%%%%%%%%%%%%%%%%%%%%%%%%%%%%%%%%%%%%%%%%%%
Standard complex quantum theory is formulated on the basis of complex Hilbert 
 spaces.
For quaternionic quantum theory one uses a similar structure while taking into 
 account the noncommutativity of the quaternions.
Accordingly, one considers the $n$-fold direct product of quaternions, denoted 
 by $\quats^n$.
It forms a free bimodule and possesses, apart from commutativity, most 
 properties of a vector space.
In particular, since it arises from a direct product, it can be equipped with 
 the canonical basis $(e^{(1)},e^{(2)},\dotsc,e^{(n)})$.

For $x,y\in \quats^n$ with $x=\sum_i x_i e^{(i)}$ and $y$ similar we define the 
 inner product to be of the canonical form $\exv{x,y}= \sum_i x^*_i y_i = 
 \exv{y,x}^*$, giving also rise to the norm $\norm{x}=\sqrt{\exv{x,x}}$ and 
 turning $\quats^n$ into a Hilbert module.
For scalar multiplication with $\alpha\in \quats$ we obtain the rules 
 $\exv{x\alpha,y}=\alpha^* \exv{x,y}$ and $\exv{x,y\alpha}= \exv{x,y}\alpha$.
This suggests that scalar multiplication in $\quats^n$ is preferably taken from 
 the right, although, technically, $\quats^n$ is a bimodule.

We take linear maps $M\colon \quats^n\to \quats^m$ to be right-homogeneous, 
 $M(x\alpha)=M(x)\alpha$ which allows for a representation of $M$ as an 
 $m\times n$ matrix $(M_{i,j})_{i,j}$ via $M_{i,j}= \exv{e^{(i)},Me^{(j)}}$.
Then $\exv{x,M(y\alpha)}= \sum_{i,j} x_i^* M_{i,j} y_j \alpha$.
We consider $\Mat(n, \quats)$ as an $\reals$-algebra and, where unambiguous, we 
 use $\alpha\in \quats$ also as the linear map $x\mapsto \sum_i \alpha x_i 
 e^{(i)}$.
The adjoint $\dag$ of a linear map is defined as usual, $\exv{x,M(y)}= 
 \exv{M^\dag(x),y}$, and therefore $(M^\dag)_{i,j}=(M_{j,i}^*)_{i,j}$.
Since linear maps are well-represented by matrices, we mostly use the latter 
 notion.
Self-adjoint and skew-adjoint matrices are defined in the obvious way.
A unitary matrix obeys $U^\dag U=\openone$ and we mention that this condition 
 is equivalent to \cite{Zhang1997} $U U^\dag = \openone$.
If $M$ is normal, $M^\dag M=M M^\dag$, then there exits a diagonal matrix $D$ 
 with entries in $\set{a+bi | a,b\in \reals,\; b\ge 0}$ and a unitary matrix 
 $U$ such that \cite{Zhang1997} $M=U^\dag D U$.
The trace $\tr(M)=\sum_i M_{i,i}$ for self-adjoint matrices is invariant under 
 unitary transformations \cite{Aleksandar1995}, $\tr(UMU^\dag)=\tr(M)$, and 
 thus equivalent to the sum of diagonal elements of $D$.

It is sometimes convenient to use the embedding $\Lambda\colon \Mat(n, \quats) 
 \to \Mat(2n, \compl)$,
\begin{equation}
 \Lambda\colon A+Bj \mapsto
 \begin{pmatrix} A & B \\ -B^* & A^* \end{pmatrix}.
\end{equation}
This map is a $\dag$-monomorphism of the corresponding $\reals$-algebras.
In particular, we have $\Lambda[(r A+BC^\dag)]= r\Lambda(A)+ 
 \Lambda(B)\Lambda(C)^\dag$ for $r\in \reals$ and $A,B,C\in \Mat(n,\quats)$.
The map $\Lambda^{-1}\colon \Mat(2n,\compl)\to \Mat(n,\quats)$,
\begin{equation}
 \Lambda^{-1}\colon \begin{pmatrix} A&B \\ C&D\end{pmatrix}
 \mapsto A+Bj,
\end{equation}
 is an $\reals$-linear left inverse of $\Lambda$, which, however, it is not 
 preserving the algebraic properties of $\Mat(2n,\compl)$.

%%%%%%%%%%%%%%%%%%%%%%%%%%%%%%%%%%%%%%%%%%%%%%%%%%%%%%%%%%%%%%%%%%%%%%%%%%%%
\subsection{Stone's theorem}
%%%%%%%%%%%%%%%%%%%%%%%%%%%%%%%%%%%%%%%%%%%%%%%%%%%%%%%%%%%%%%%%%%%%%%%%%%%%
In order to study the dynamics in quaternionic quantum theory we proceed 
 similarly to the complex case by studying continuous unitary groups $(U_t)_t$.
Then, an analogous result to Stone's theorem holds.
\begin{theorem}\label{thm:qstone}
For every continuous unitary group $(U_t)_{t \in \reals} \subset \Mat(n, 
 \quats)$ there exists a unique skew-adjoint matrix $A\in \Mat(n, \quats)$ such 
 that $U_t = \ee^{At}$.
\end{theorem}
This theorem has been proved in Ref.~\onlinecite{Finkelstein1962}.
For completeness, we provide here a proof for finite dimensions.
\begin{proof}[Proof of Theorem~\ref{thm:qstone}]
Since the embedding $\Lambda$ is a mapping between finite-dimensional vector 
 spaces, the family $[\Lambda(U_t)]_t\subset \Mat(2n,\compl)$ is also a 
 continuous unitary group.
By virtue of Stone's theorem there exists a skew-adjoint matrix $B\in 
 \Mat(2n,\compl)$ such that $\Lambda(U_t)= \ee^{Bt}$.
The map $\mathcal W\colon t\mapsto \Lambda(U_t)$ is also differentiable, so 
 that we can write for $\mathcal U\colon t\mapsto U_t$,
\begin{equation}
 \Lambda^{-1}(\dot{\mathcal W})
 = \frac{\dd}{\dd t} \Lambda^{-1} (\mathcal W)
 = \dot{\mathcal U}.
\end{equation}
The left hand side exists, proving that also $\mathcal U$ is differentiable.
By letting $A= \dot{\mathcal U}(0)$ and using $U_{t+\delta}= U_\delta U_t$, we 
 have
\begin{equation}
 \dot{\mathcal U}(t) = \lim_{\delta \to 0}
 \frac{U_{t + \delta} - U_t}{\delta}= AU_t.
\end{equation}
It remains to show that $A$ is skew-adjoint and satisfies $U_t=\ee^{At}$.
Since $\Lambda(U_t)$ is unitary, the identity $\Lambda(A)\Lambda(\mathcal U)= 
 \Lambda(A\mathcal U) = \Lambda(\dot{\mathcal U})=\frac\dd{\dd t} 
 \Lambda(\mathcal U) = \dot{\mathcal W}= B\mathcal W = B\Lambda(\mathcal U)$, 
 allows us to conclude that $\Lambda(A)=B$.
This implies that $A$ is skew-adjoint, since $B$ is.
Finally, applying $\Lambda^{-1}$ to
\begin{equation}
 \Lambda(U_t)= \ee^{Bt}= \ee^{\Lambda(A)t}=\Lambda(\ee^{At})
\end{equation}
 from the left gives us $U_t=\ee^{At}$.
Uniqueness then follows immediately from $\dot{\mathcal U}(0)=A$.
\end{proof}
We mention that the smoothness of the map $t\mapsto U_t$, which we show in the 
 first part of the above proof, is a simple consequence of the fact that the 
 unitary matrices form a Lie group.
Indeed, for any Lie group $G$ a continuous homomorphism $\reals \to G$ is 
 necessarily smooth \cite{Lee2012}.
Also note that in contrast to the complex case, we consider here only the 
 finite-dimensional case and therefore it is not necessary to use strong 
 continuity.

%%%%%%%%%%%%%%%%%%%%%%%%%%%%%%%%%%%%%%%%%%%%%%%%%%%%%%%%%%%%%%%%%%%%%%%%%%%%
\subsection{Hamiltonians and generators of time shifts}
%%%%%%%%%%%%%%%%%%%%%%%%%%%%%%%%%%%%%%%%%%%%%%%%%%%%%%%%%%%%%%%%%%%%%%%%%%%%
We now head for the characterization of the dynamics in a quaternionic version 
 of quantum theory.
So far we have obtained a result about the structure of all possible dynamical 
 evolutions.
But for a dynamical evolution to be useful we need to specify a notion of 
 states and observables.
Here, we proceed in complete analogy to quantum theory, that is, states are 
 represented by normalized vectors and observables by self-adjoint matrices 
 \cite{Adler1995}.
The expectation value of an observable $H$ for a system in state $\psi$ is 
 defined as $\exv{H}= \exv{\psi,H\psi}$.
This yields always a real number due to $\exv{H}=\exv{H\psi,\psi}^*= \exv{\psi, 
 H^\dag \psi}^*= \exv{H}^*$, and all states $\psi\alpha$ are equivalent for all 
 $\alpha\in \quats$ with $\abs\alpha= 1$.
The spectral theorem for self-adjoint matrices can also be written as $H= 
 \sum_k h_k \Pi_k$, with distinct eigenvalues $h_k\in \reals$ and self-adjoint 
 projections $\Pi_k$ such that $\sum_k\Pi_k=\openone$.
The expectation values of the projections correspond then to the probability 
 $p_k$ for observing the eigenvalue $h_k$, that is, $p_k=\exv{\Pi_k}$.
In this way we recover a large bit of the structure and physical interpretation 
 of quantum theory.

In analogy to the complex case, we assume that the time evolution of a state is 
 generated by a continuous unitary group $(U_t)_t$ and by virtue of 
 Theorem~\ref{thm:qstone} we have $U_t=\ee^{At}$.
We consider now the correspondence between Hamiltonians $H$ and generators of 
 time shifts $A$.
It is worth mentioning that also continuous groups of nonlinear isometries are 
 possible.
For example, such a group can be constructed from a continuous unitary group 
 $(U_t)_t$ as $(R[U_t])_t$, where $R[M]\colon x\mapsto \sum_{i,k} x_i 
 (M^\dag)_{i,k} e^{(k)}= (M x^*)^*$ is an additive map that is not 
 right-homogeneous.
In order to maintain the analogy to standard complex quantum theory, here we do 
 not consider such nonlinear isometries.

In the complex case, the multiplication with a purely imaginary number 
 $i\lambda$ is the right choice to establish this correspondence between 
 Hamiltonian $H$ and generator $A$.
The generator $A$ can written in the polar decomposition as 
 \cite{Finkelstein1962} $A=-XH$, where $X$ is unitary and skew-adjoint, $H$ is 
 self-adjoint and positive semidefinite, and $[X,H]= 0$ holds.
It is conceivable to identify $H$ in this decomposition with the Hamiltonian of 
 the system while $X$ is kept constant.
However, this limits the possible set of Hamiltonians to those with $[X,H]= 0$ 
 which basically reduces quaternionic quantum theory to standard complex 
 quantum theory \cite{Finkelstein1962}.

Here, we are interested in the case where the Hamiltonian of the system can be 
 any self-adjoint operator.
The discussion in Section~\ref{sec:cgens} remains valid and leaves us with the 
 task to characterize the commuting $\reals$-linear maps $\Phi\colon 
 \Herm(n,\quats) \to \AHerm(n,\quats)$.
However, we cannot proceed similarly to above to obtain a result akin to 
 Theorem~\ref{thm:cphi}.
The main difficulty here is that it is not possible to use an extension of 
 $\Phi$ as in Section~\ref{sec:cgens}, since such an extended map would be no 
 longer commuting.
We hence resort to a case-by-case study for different dimensions $n$.

For this it is useful to note that determining all admissible maps $\Phi$ can 
 be reduced to finding the kernel of an $\reals$-linear map.
Indeed, $\Phi$ is commuting if and only if the $\reals$-bilinear map $\mathcal 
 Q_\Phi\colon (X,Y)\mapsto [\Phi(X),Y]+[\Phi(Y),X]$ is trivial for all $X,Y\in 
 \Herm(n,\quats)$, that is, $\mathcal Q_\Phi=0$.
Here, sufficiency follows immediately for $Y=X$ and necessity from 
 $[\Phi(X+Y),X+Y]=0$.
Thus, the set of commuting maps $\Phi$ is determined by the kernel of the 
 $\reals$-linear map $\Phi\mapsto \mathcal Q_\Phi$.
In more detail, to compute this kernel first observe that $[ \Herm(n,\quats), 
\AHerm(n,\quats) ] \subset \Herm( n, \quats) $.
Fixing a basis $(t_{i})_{i}$ and $(s_a)_a$ of the real vector spaces of 
 $\Herm(n, \quats)$ and $\AHerm(n, \quats)$, respectively, allows us to write 
 $[s_a, t_i] = \sum_k C_{a,i}^k t_{k}$ with real coefficients $C_{a,i}^k$ and 
 $\Phi(t_\ell)=\sum_{a} P^\Phi_{\ell,a}s_a$ with real coefficients 
 $P^\Phi_{\ell,a}$.
We define the matrix
\begin{equation}
  \mathcal X_{(i,j,k),(\ell,a)}
  = C_{a,j}^k\delta_{i,\ell}+C_{a,i}^k\delta_{j,\ell}
\end{equation}
 with row index $(i,j,k)$, column index $(\ell,a)$ and $\delta_{i,j}$ denoting 
 the Kronecker delta.
Then $\mathcal Q_\Phi(t_i,t_j)=\sum_{k,\ell,a} \mathcal X_{(i,j,k),(\ell,a)} 
 P^\Phi_{\ell,a} t_k$ and therefore the matrix $\mathcal X$ corresponds to the 
 map $\Phi\mapsto Q_\Phi$.

We compute the matrix $\mathcal X$ and its kernel with the help of computer 
 algebra for $n=1,2,3$ with the following results.
For $n=1$, all admissible maps are (obviously) given by $\Phi_1(H)= \alpha H$ 
 for any $\alpha\in \quats$ with $\alpha = -\alpha^*$.
For $n=2$, all admissible maps are of the form $\Phi_2(H)= AH+HA-\tr(H)A$, 
 where $A\in \AHerm(2,\quats)$ is arbitrary.
For $n=3$, only the trivial map $\Phi=0$ is commuting.
This implies also that for $n>3$ all commuting maps must be trivial, as it 
 follows from the following observation.
\begin{lemma}
Any nontrivial commuting $\reals$-linear map $\Phi\colon \Herm(n,\quats)\to 
 \AHerm(n,\quats)$ with $n> 3$ induces a nontrivial commuting $\reals$-linear 
 map $\Phi_3\colon \Herm(3,\quats)\to \AHerm(3,\quats)$.
\end{lemma}
\begin{proof}
The real span of the matrices of the form $(u_iu_j^*)_{i,j}$ with $u\in 
 \quats^n$ is $\Herm(n,\quats)$.
Hence, given a nontrivial map $\Phi$, we can choose linearly independent 
 vectors $x, y, z\in \quats^n$ such that $\exv{y, \Phi(X) z}\ne 0$ with 
 $X_{i,j}=x_i x_j^*$.
Then there is a linear isometry $\tau\colon \quats^3\to \quats^n$ such that 
 $x,y,z\in \tau\quats^3$ and $\pi=\tau\tau^\dag$ acts as identity on $x$, $y$, 
 and $z$.
Such an isometry can be constructed by means of the Gram--Schmidt procedure, 
 yielding orthonormal vectors which are then used as columns of the matrix of 
 $\tau$.
We define the map $\Phi_3$ as $\Phi_3(H)= \tau^\dag \Phi(\tau H\tau^\dag)\tau$.
By construction, this map is $\reals$-linear and maps self-adjoint matrices to 
 skew-adjoint matrices.
Since $\Phi$ is commuting, we have $[ \Phi(\tau H\tau^\dag), \tau 
 H\tau^\dag]=0$ for any $H\in \Herm(3,\quats)$ and by multiplication with 
 $\tau^\dag$ from the left and $\tau$ from the right, it follows that $\Phi_3$ 
 is also commuting.
Finally, $\Phi$ is nontrivial for $X'=\tau^\dag X\tau$, $y'=\tau^\dag y$, and 
 $z'=\tau^\dag z$ in that we have $\exv{y',\Phi_3(X')z'}= \exv{\pi y, \Phi(\pi 
 X\pi)\pi z} =\exv{y,\Phi(X)z}\ne 0$, where $\pi X\pi=X$ follows from $\pi X 
 \pi v= \pi x\exv{x,\pi v}= x\exv{\pi x, v}= Xv$ for all $v\in \quats^n$.
\end{proof}

In summary we have the following characterization.
\begin{theorem}\label{thm:qrep}
For $n=1,2$, every commuting $\reals$-linear map $\Phi\colon \Herm(n,\quats)\to 
 \AHerm(n,\quats)$ is of the form
\begin{equation}
 \Phi(H)= AH+HA-\tr(H)A,
\end{equation}
 where $A\in \AHerm(n,\quats)$.
Conversely, for $n=1,2$ every map of this form is commuting.
For $n>2$ every commuting $\reals$-linear map is trivial, $\Phi=0$.
\end{theorem}

For a two-level system, $n=2$, the rank of $\Phi$ can be at most
 $\dim[\Herm(n,\quats)]-1=5$ due to $\Phi(\openone)=0$, for any choice of $A$.
The maximal rank is achieved, for example, for $A$ being the diagonal matrix 
 with diagonal $(i,0)$.
A more intuitive choice for $A$ might be $-\frac i{2\hbar}$, which yields the 
 quaternionic Schrödinger-type equation
\begin{equation}
 i\hbar\dot\psi(t)= [H_c- \tfrac12\tr(H_c)\openone]\psi(t),
\end{equation}
 where $H_c=\frac12(H+iHi^*)$ is the complex part of the matrix $H$.
Hence the corresponding map $\Phi$ has an $\reals$-rank of only 3 and the 
 dynamics reduces to the well-known case of two-level quaternionic systems 
 where all time evolution is generated solely by the complex 
 Hamiltonians\cite{Brody2011}.

Theorem~\ref{thm:qrep} has been obtained using the conditions (DC1)--(DC3).
In contrast to the complex case, the resulting map $\Phi$ only obeys the 
 covariance condition (DC4) if it is trivial.
This can be seen by defining $\Phi_V\colon H\mapsto V\Phi(V^\dag HV)V^\dag= 
 A_VH+HA_V-\tr(H)A_V$ with $A_V=V A V^\dag$.
Covariance requires then $\Phi_V=\Phi$ for any unitary matrix $V$.
Without loss of generality, we can choose $A=iD$ with $D$ a real diagonal 
 matrix and then the case $V=j$ leads to $\Phi_V=-\Phi$.

However, the no-go statement in Theorem~\ref{thm:qrep} can be avoided by 
 loosening our assumptions.
For example, one can drop the assumption of additivity (DC2).
Then a rather natural such candidate can be achieved as follows.
We fix a spectral decomposition $H=U_H^\dag D_H U_H$ for every $H$ such that 
 $U_{rH}$ is independent of $r\in \reals$.
This allows us to define the $\reals$-homogeneous map $\Psi\colon 
 \Herm(n,\quats)\to \AHerm(n, \quats)$ as
\begin{equation}
 \Psi\colon H\mapsto U_H^\dag i D_H U_H,
\end{equation}
 which can be easily seen to be commuting, but, in general, fails to be 
 additive, due to Theorem~\ref{thm:qrep}.
We can even satisfy the covariance condition (DC4) by requiring that 
 $U_{VHV^\dag}=VU_H$ for all unitaries $V$.
As mentioned before, another way to evade Theorem~\ref{thm:qrep} is to limit 
 the set of Hamiltonians to be purely complex \cite{Finkelstein1962}, $[H,i]= 
 0$.
Then $\Phi_S\colon H\mapsto -\frac i\hbar H$ is admissible under (DC1)--(DC3) 
 and we obtain the Schrödinger equation also in the quaternionic case.

%%%%%%%%%%%%%%%%%%%%%%%%%%%%%%%%%%%%%%%%%%%%%%%%%%%%%%%%%%%%%%%%%%%%%%%%%%%%
\section{Conclusions}\label{sec:conclusions}
%%%%%%%%%%%%%%%%%%%%%%%%%%%%%%%%%%%%%%%%%%%%%%%%%%%%%%%%%%%%%%%%%%%%%%%%%%%%
We studied the structure of universal dynamics in quantum theory using three 
 main axioms (DC1)--(DC3).
These axioms prove to be sufficient in order to recover the Schrödinger 
 equation for the case of standard complex quantum theory but when applied to 
 quaternionic quantum theory they yield nontrivial dynamics only for dimension 
 one and two.
For two-level systems, the resulting Schrödinger-type equation is not unique 
 but can be modified by the choice of a skew-adjoint operator, see 
 Theorem~\ref{thm:qrep}.

This makes quaternionic quantum theory for two-level systems exceptionally 
 interesting.
For higher dimensions, a possible conclusion from our analysis is to discard
 quaternionic quantum theory.
However, it should be noted that the main reason for our no-go result is axiom 
 (DC2), which requires that the sum of two Hamiltonians should correspond to 
 the sum of two generators of time shifts.
While this axiom is natural at least in standard complex quantum theory it is 
 an open question whether it is expendable nonetheless, and what this would 
 imply for complex quantum theory as well as quaternionic quantum theory.

\begin{acknowledgments}
This work was supported by
the DFG,
the FQXi Large Grant “The Observer Observed: A Bayesian Route to the 
Reconstruction of Quantum Theory”, and
the ERC (Consolidator Grant 683107/TempoQ).
\end{acknowledgments}

\bibliography{the}

\end{document}